\documentclass[final, 10pt, twocolumn, conference]{IEEEtran}
\usepackage[left=0.75in,right=0.75in,top=0.75in,bottom=0.75in]{geometry}
\usepackage{setspace}

\usepackage{amsmath}
\usepackage{amsthm}
\usepackage{amssymb}
\usepackage{cite}
\usepackage{graphicx}
\usepackage{epstopdf}
\usepackage{url}
\usepackage{cite}
\usepackage{texdef2015}
\usepackage{color}
\usepackage{caption}
\usepackage{subcaption}
\usepackage{float}
\usepackage{balance}
\allowdisplaybreaks

\newtheorem{theorem}{Theorem}
\newtheorem{corollary}{Corollary}

\newcommand{\age}{\Delta}

\newcommand{\negfigspace}{\vspace{0mm}}

\newcommand{\cf}{c}

\newcommand{\wqr}{k}

\newcommand{\wset}{\mathcal{K}}

\usepackage{tikz}
\usetikzlibrary{automata,arrows,positioning,calc,fit,shapes.multipart,chains,shapes}
\usepackage{microtype}

\begin{document}
\title{\vspace{3.5ex}Status Updates Through Multicast Networks}
\author{Jing Zhong, Emina Soljanin and Roy D.~Yates \\
Department of ECE, Rutgers University, \{jing.zhong, emina.soljanin, ryates\}@rutgers.edu}


\maketitle

\begin{abstract}
Using age of information as the freshness metric, we examine a multicast network in which real-time status updates are generated by the source and sent to a group of $n$ interested receivers.
We show that in order to keep the information freshness at each receiver, the source should terminate the transmission of the current update and start sending a new update packet as soon as it receives the acknowledgements back from any $\wqr$ out of $n$ nodes. 
As the source stopping threshold $\wqr$ increases, a node is more likely to get the latest generated update, but the age of the most recent update is more likely to become outdated.
We derive the age minimized stopping threshold $\wqr$ that balances the likelihood of getting the latest update and the freshness of the latest update for shifted exponential link delay. 
Through numerical evaluations for different stopping strategies, we find that waiting for the acknowledgements from the earliest $\wqr$ out of $n$ nodes leads to lower average age than waiting for a pre-selected group of $\wqr$ nodes.
We also observe that a properly chosen threshold $\wqr$ can prevent information staleness for increasing number of nodes $n$ in the multicast network. 
\end{abstract}

\section{Introduction}

Recent advance in pervasive connectivity and ubiquitous computing has engendered many applications that requires real-time information updates, including environmental sensor networks monitoring temperature and humidity, and connected vehicular network by self-driving cars.
These applications share a common need: the freshness of the data at the interested recipients has to be maximized. 

Recent works on status updating system has been focused on the analysis of a new  ``Age of Information" (AoI) timeliness metric  \cite{Kaul2012infocom,Costa2014,Huang2015,Sun2016,Najm2017,Bedewy2016,Kadota2016,Yates2017}.
In most status updating systems, a source generates time-stamped status update messages that are transmitted through a communication system to a receiver. 
Age of information, or simply \emph{age}, measures the time difference between now and when the most recent update was generated.
If the receiver receives an update at some time $t$, and the update packet was generated at time $u(t)$, then the instantaneous age of the update packet is $t-u(t)$.

In next-generation wireless networks with pervasive computing and dense IoT deployments, most delay-sensitive data will be popular and simultaneously requested by large numbers of users.
For example, the information updates of an autonomous car will be broadcast to all nearby vehicles and passengers equipped with portable sensing devices, and the live video captured by a sports motion camera may be streamed to all portable devices in a stadium. 
In order to better utilize the radio spectrum and reduce the cost of resources, the 3GPP proposes to deliver multimedia content to the mobile devices using the evolved Multimedia Broadcast and Multicast Service (eMBMS) \cite{Lecompte2012,Lentisco}. 
For time-sensitive data gathering and monitoring systems consisting of multiple recipients, we are interested in the problem of how should the source send update messages through multicast or broadcast channels such that the age of information at each receiver can be minimized. 

Although the age of information 
has been characterized in a variety of contexts by recent works, the study of status updating through multicast or broadcast networks is very limited.
In \cite{Kadota2016}, it was shown that a greedy scheduling policy that prioritize the highest age packet provides optimal weighted sum of age over all users for wireless broadcast networks.
In \cite{Yates2017}, two different hybrid ARQ schemes were compared for status updating over multiple binary erasure channels, and showed that the average age at each user depends on the total number of users in the system.

In this work, we examine status updating systems in which real-time status update messages generated by the source are sent to a set of nodes through multicast channels with random network delays.
Once a node receives an update message, it acknowledges the source by sending a response through an instantaneous feedback channel.
Assuming that the source can terminate the transmission of the current update at any time and start the transmission of the next update, 
we ask the following question: how long should the source wait until it starts the next update in a multicast system with $n$ nodes?
We start by considering a simple strategy that the source terminates the current update and sends out a new update as soon as it receives instantaneous responses from the earliest $\wqr$ nodes.
Intuitively, if the source waits for more responses, each node is more likely to get an update, but the update message is also more likely to be outdated. 
We show that the optimal stopping threshold $\wqr$, which balances the likelihood of delivering an update and the freshness of the most recent update, depends on the order statistics of the delay distributions in the multicast network.
We also note that the model we consider is relevant to many systems, e.g., coded multicast, where single source transmits coded file packets to multiple servers  \cite{HeindlmaierS14}.
In these schemes, it is often required that each packet reaches a certain minimum number of users to ensure that the packet loss that each user experiences is recoverable by the code.

In Sec.~\ref{sec:system}, we formulate the system model of updating through multicast channels.
We then analyze the time-averaged age at each node in Sec.~\ref{sec:single}. We show that for i.i.d. shifted-exponential delay at each link, an approximate average age can be obtained by exploiting the logarithmic approximation of the order statistics, and there exists an optimal stopping threshold $\wqr$ that minimizes the average age at each node.
We also investigate another stopping scheme in which the source selects a group of $\wqr$ nodes randomly in advance, and only waits for the responses from this group of $\wqr$ nodes.
The numerical evaluations in Sec.~\ref{sec:evaluation} demonstrates how the choice of threshold $\wqr$ affects the average age for different distributions, and proves the tightness of our approximation.
We conclude by summarizing our results and discussing possible future work in Sec.~\ref{sec:conclusion}.

\section{System Model and Metric}\label{sec:system}
\begin{figure}[t]
\centering\small
\begin{tikzpicture}[node distance=1cm]
\node [draw,circle, rounded corners,align=center] (newsource) {Source};
\node[draw,rectangle,rounded corners,align=center,minimum height=4mm,thick] (node_2)[right = 4 of newsource] {Node 2};
\node[draw,rectangle,rounded corners,align=center,minimum height=4mm,thick] (node_1)[above = 0.1 of node_2] {Node 1};
\node[draw,rectangle,rounded corners,align=center,minimum height=4mm,thick] (node_3)[below = 0.1 of node_2] {Node 3};
\node[draw,rectangle,rounded corners,align=center,minimum height=4mm,thick] (node_n)[below = 1.2 of node_2] {Node n};
\draw[->,thick] (newsource.east) -- node[draw,rectangle,minimum height=3mm,thin][above left]{$j$+1} ++(1.5,0) |-  (node_1.west) node[draw,rectangle,minimum height=3mm,thin] [above left = 0 and 0.2]{$j$};
\draw[->,thick] (newsource.east) -- ++(1.5,0) |- (node_3.west)node[draw,rectangle,minimum height=3mm,thin,opacity=.6] [above left = 0 and 0.7]{$j$};
\draw[->,thick] (newsource.east) -- (node_2.west) node[draw,rectangle,minimum height=3mm,thin,opacity=.6] [above left = 0 and 1]{$j$}; 
\draw[->,thick] (newsource.east) -- ++(1.5,0) |- (node_n.west) node[draw,rectangle,minimum height=3mm,thin] [above left = 0 and 0.3]{$j$};
\path (node_3.south) -- node {$\cdots$}  (node_n.north);
\draw[red,dashed,rounded corners] (4.5,1.2) node [right] {$\wqr$ nodes have update $j$} rectangle (3,-1.8) ; 

\end{tikzpicture}
\caption{Source broadcasts status updates to multiple nodes through i.i.d. channels. The transmission of update $j+1$ is initiated after update $j$ is delivered to $\wqr$ out of $n$ nodes.}
\label{fig:sysmodel}
\negfigspace
\end{figure}
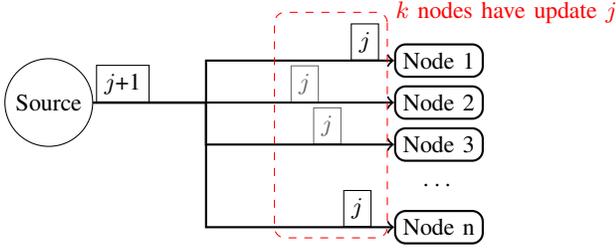

We consider a system with a single source broadcasting time-stamped updates to $n$ nodes through $n$ links with independent random delays, as shown in Fig.~\ref{fig:sysmodel}.
Each update is time-stamped when it was generated.
An update takes time $X_i$ to reach node $i$. We refer to $X_i$ as the service time of link $i$.
We assume that the $X_i$ are i.i.d. shifted exponential $(\lambda,c)$ random variables. Consequently, each $X_i$ has CDF $F_{X}(x)~=~1-e^{-\lambda(x-c)}$, for $x\geq c$.
The constant time shift $c>0$ captures the delay produced by the update generation and assembly process.
On the other hand, it could also be a propagation delay on top of an exponential network delay if the source and database are geographically separated.

When most recently received update at time $t$ at node $i$ is time-stamped at time $u_i(t)$, the status update age or simply the {\it age}, is the random process $\age_i(t)=t-u_i(t)$. 
When an update reaches node $i$, $u_i(t)$ is advanced to the timestamp of the new update message and the node sends an acknowledgement to the source through an instantaneous feedback channel. 
When the earliest $\wqr$ out of $n$ nodes report receiving the update $j$, we say that update $j$ has been completed. 
At this time, the transmissions of all $n-\wqr$ remaining replicas are terminated. 
Following the completion of update $j$, we could choose to insert a waiting time before starting the transmission of update $j+1$ as it has been shown that a nonzero wait can reduce the age in  certain single node systems \cite{Sun2016}.  However, in this work we restrict our attention to zero-wait policies in which update $j+1$ begins transmission immediately following the completion time of update $j$. 
We refer to this stopping scheme as the \emph{earliest~k}.
Alternatively, the source can choose the $\wqr$ out of $n$ nodes as a group in advance, and only wait for the acknowledgements from this group of nodes. We call this stopping scheme the \emph{pre-selected~k}.
Note that $\wqr=n$ in the special case that the source starts the transmission of the next update only if the current update is delivered to all $n$ nodes, which is called \emph{wait-for-all}.

The time average of age process is defined as the age of information
\begin{align}
\age = \lim_{\tau\to\infty} \frac{1}{\tau} \int_{0}^{\tau} \age(t).
\end{align}
We will use a graphical argument similar to that in \cite{Yates2017} to evaluate the average age at an individual node and at the client.
Fig.~\ref{fig:sawtooth_n_1} depicts a sample path of the age over time at some node $i$ in a system of $n$ nodes.
Update $1$ begins transmission at time $t=0$ and is timestamped $T_0=0$. 
Here we define $X_{ij}$ as service time to deliver the update $j$ to node $i$.
Since $\{X_{ij}\}$ are i.i.d. for all $i$ and $j$, the $\age_i(t)$ processes are statistically identical and each node $i$ will have the same average age $\age_i$.

\section{Age Analysis} \label{sec:single}

\begin{figure}[t]
\centering
\begin{tikzpicture}[scale=0.25]
\draw [fill=lightgray, ultra thin, dashed] (0,0) to (0,4) to (3,7) to (3,3) to (7,7) to (7,0);
\draw [fill=lightgray, ultra thin, dashed] (17,0) to (17,6) to (20,9) to (20,3) to (23,6) to (23,0);
\draw [<-|] (0,12) node [above] {$\age_{(n)}(t)$} -- (0,0) -- (13,0);
\draw [|->] (14,0) -- (28,0) node [right] {$t$};
\draw 
(3,0) node {$\bullet$}  
(7,0) node [below] {$T_1$}
(11,0) node {$\bullet$}
(17,0) node [below] {$T_{j-1}$} 
(20,0) node {$\bullet$} 
(23,0) node [below] {$T_{j}$}
(25,0) node {$\bullet$};
\draw  [<-] (3,2) to [out=60,in=290] (5,7) node [above] {$A_1$};
\draw[<-] (20,2) to [out=60,in=280] (22,7) node [above] {$A_{j}$};
\draw [very thick] (0,4) -- (3,7) -- (3,3)  -- (7,7) 
-- (11,11)-- (11,4) --(13,6);
\draw [very thick] (14,3) -- (20,9) -- (20,3)  -- (25,8) -- (25,5) -- (28,8); 
\draw  [|<->|] (0,10) to node [above] {$Y_{1}$} (7,10);
\draw  [|<->|] (17,10) to node [above] {$Y_{j}$} (23,10);
\draw  [|<->|] (0,-2.5) to node [below] {$X_{i1}$} (3,-2.5);
\draw  [|<->|] (7,-2.5) to node [below] {$X_{i2}$} (11,-2.5);
\draw  [|<->|] (17,-2.5) to node [below] {$X_{ij}$} (20,-2.5);
\end{tikzpicture}
\caption{Sample path of the age $\age_{(n)}(t)$ for node $i$ with $n$ nodes. Update delivery instances are marked by $\bullet$.}
\label{fig:sawtooth_n_1}\negfigspace
\end{figure}
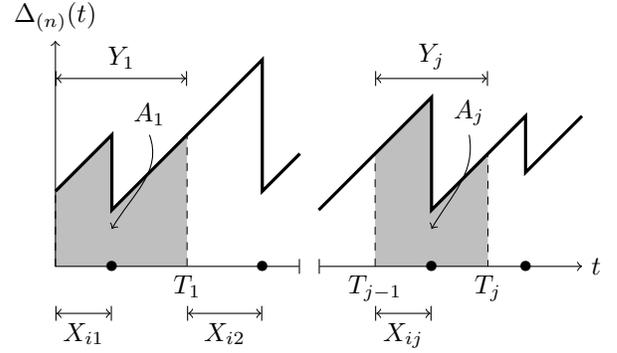  
\subsection{Order Statistics Notation}
We denote the $\wqr$-th order statistic of
random variables $X_1, \ldots, X_n$, i.e., the $\wqr$-th smallest variable, as $X_{k:n}$.
For shifted exponential $X$, the expectation and variance of the order statistic $X_{k:n}$ are given by 
\begin{align}
\E{X_{k:n}} & = \cf+\frac{1}{\lambda}(H_n-H_{n-k}) \label{os_first} \\
\Var{X_{k:n}} & = \frac{1}{\lambda^2}\left(H_{n^2}-H_{(n-k)^2}\right) \label{os_var}
\end{align}
where $H_n$ and $H_{n^2}$ are the generalized harmonic numbers defined as $H_n = \sum_{j=1}^{n}\frac{1}{j}$ and $H_{n^2} = \sum_{j=1}^{n}\frac{1}{j^2}$. 
Thus the second moment of the order statistic is 
\begin{align}
\E{X_{k:n}^2} & = \cf^2+\frac{2c}{\lambda}(H_n-H_{n-k}) \nn
& \quad + \frac{1}{\lambda^2}\left((H_n-H_{n-k})^2 + H_{n^2}-H_{(n-k)^2}\right). \label{os_second}
\end{align}
\subsection{Wait-for-all Scheme}
We first consider the special case where the source waits for the receiving acknowledgements from all $\wqr=n$ nodes; that is, the delivery of an update to all $n$ nodes is guaranteed. 
If one node gets an update earlier than any of the other $n-1$ nodes, it has to wait for an idle period until that the update is delivered to all $n$ nodes.
The transmission time of an update $j$ to all $n$ nodes, which we call a \emph{service interval}, is given by
\begin{align}
Y_j&=\max(X_{1j}, \ldots, X_{nj}).
\end{align}
Note that $Y_j$ is identical to $X_{n:n}$.
We denote that update $j$ goes into service at time $T_{j-1}$ with time stamp $T_{j-1}$ and get delivered to all $n$ nodes at time $T_j = T_{j-1} + Y_j$.
As shown in Figure \ref{fig:sawtooth_n_1}, the age $\age_i(t)$ drops when update $j$ is delivered to node $i$ at an earlier time $T_{j-1}+X_{ij}$. 
This implies at time $T_j$ when update $j$ completes transmission to all nodes, the age at monitor $i$ is $\age_i(T_i)=Y_j$.
We define $\age_{(n)}$ as the average age $\age_i$ at some node $i$ for the wait-for-all scheme with $n$ nodes in total.
We represent the area under the age sawtooth as the concatenation of the polygons $A_1,\ldots,A_j$,  thus the average age is
$\age_{(n)}=\E{A}/\E{Y}$.
Fig.~\ref{fig:sawtooth_n_1} shows that
\begin{align}
A_i&=Y_{j-1}X_{ij} +X_{ij}^2/2 +X_{ij}(Y_j-X_{ij})+(Y_j-X_{ij})^2/2\nn
&=Y_{j-1}X_{ij}+Y_j^2/2.
\end{align}
Since $X_{ij}$ is independent of the transmission time $Y_{j-1}$ of the previous update, $\E{A}=\E{Y}\E{X}+\E{Y^2}/2$ and we have the following theorem
\begin{theorem} For wait-for-all stopping scheme with $n$ nodes, the average age at an individual node is 
\begin{align*}
\age_{(n)} & = \E{X} + \frac{\E{X_{n:n}^2}}{2\E{X_{n:n}}}.
\end{align*}
\label{thm:age_n1}
\end{theorem}
We remark that Theorem \ref{thm:age_n1} holds for any distribution of $X$ if zero-wait policy is used and the update delivery acknowledgement is instantaneous.
\begin{corollary} For shifted exponential $(\lambda,c)$ service time $X$, the average age at an individual node for wait-for-all stopping scheme is given by
\begin{align*}
\age_{(n)} & = \frac{3c}{2} + \frac{1}{\lambda} + \frac{H_n}{2\lambda} + \frac{H_{n^2}}{2\lambda^2\cf+2\lambda H_n}.
\end{align*}  \label{thm:age_n1_exp}
\end{corollary}
Corollary \ref{thm:age_n1_exp} follows by substituting (\ref{os_first}) and (\ref{os_second}) into Theorem~\ref{thm:age_n1}. 
Note that $H_{n^2} = \sum_{i=1}^{n} \frac{1}{n^2}$ is monotonically increasing for $n\in\mathbb{Z^+}$ and $\lim_{n\to\infty} H_{n^2}=\pi^2/6$. Thus, given $\lambda$ and $\cf$,
\begin{align*}
\lim_{n\to\infty}~\frac{H_{n^2}}{2\lambda^2\cf+2\lambda H_n}~\to~0,
\end{align*}
and the average age can be approximated by 
\begin{align}
\age_{(n)} \approx \frac{3c}{2} + \frac{1}{\lambda} + \frac{H_n}{2\lambda}.
\end{align}
When $n$ is large, we can further approximate the harmonic number by $H_n\approx\log n+\gamma$ where $\gamma\approx 0.577$ is the Euler-Mascheroni constant. This implies that the average age with wait-for-all stopping strategy $\age_{(n)}$ behaves similarly to a logarithmic function as the number of nodes $n$ increases. 

\begin{figure}[t]
\centering
\begin{tikzpicture}[scale=0.095]
\draw [<-|] (0,37) node [above] {$\age_{(\wqr)}(t)$} -- (0,0) -- (15,0);
\draw [|->] (18,0) -- (80,0) node [right] {$t$};
\draw (0,0) (11,0) -- +(0,-1) (30,0) -- +(0,-1) (40,0) -- +(0,-1) (48,0) -- +(0,-1) (60,0) -- +(0,-1) (70,0) -- +(0,-1) ;
\fill[lightgray] (30,0) to ++(10,10) to ++(25,25) to ++(0,-30)  to ++(-5,-5);
\draw (7,0) node {$\bullet$}
(25,0) node {$\bullet$}
(38,0) node {$\bullet$}
(65,0) node {$\bullet$};
\draw (11,-6) node [above] {$T_1$} 
(30,-6) node [above] {$T_{j-1}$} (40,-6) node [above] {$T_{j}$} (48,-6) node [above] {$T_{j+1}$} (60,-6) node [above] {$T_{j+2}$} (70,-6) node [above] {$T_{j+3}$};
\draw [very thick] (0,10) -- ++(7,7) -- ++(0,-10) -- ++(8,8);
\draw [very thick]  (18,18) -- ++(7,7)  -- ++(0,-20) -- ++(7,7) -- (35,15) -- ++(3,3) -- ++(0,-10) -- ++(2,2) -- ++(25,25) -- ++(0,-30) -- ++(5,5);
\draw  [|<->|] (30,7) to node [below] {$Y_{j}$} (40,7);
\draw  [|<->|] (40,7) to node [below] {$Y_{j+1}$} (48,7);
\draw  [|<->|] (48,7) to node [below] {$Y_{j+2}$} (60,7);
\draw [|<-] (60,20) to (57,20);
\draw [|<-] (65,20) to (68,20) node [right] {$\Xtil$};
\draw  [|<->|] (72,0) to node [right] {$\Xtil$} (72,5);
\end{tikzpicture}
\caption{\small Sample path of the age $\age_{(\wqr)}(t)$: successful update deliveries (at times marked by $\bullet$) occur in intervals $1$, $j-1$, $j$, and $j+3$. Updates are preempted in intervals $j+1$ and $j+2$.}
\label{fig:sawtooth_w_1}
\negfigspace
\end{figure}
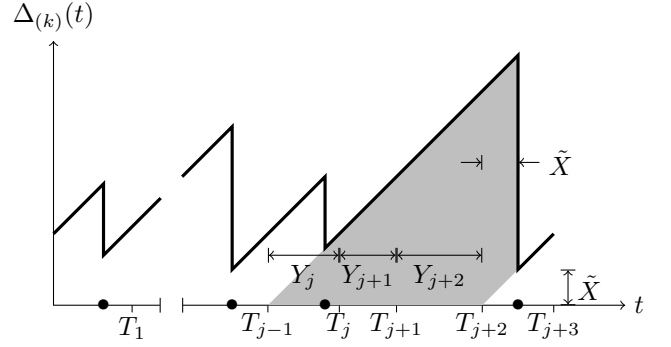
\subsection{Stopping at Earliest $\wqr$}
For earliest $\wqr$ stopping scheme, the source preempts the current update in service with a new update after the delivery of the current update to the earliest $\wqr$ out of $n$ nodes, where $\wqr\in\{1,2,\ldots, n\}$. 
In this case, the transmission time of an update $j$ to $\wqr$ nodes is $Y_j=X_{\wqr:n}$.
Denote the set of the $\wqr$ out of $n$ link with smallest service times as $\wset$.
Since the $X_i$ are i.i.d., the probability that link $i$ is the set $\wset$ is 
\begin{align}
p = \Pr[i\in\wset] = \frac{\wqr}{n}.
\end{align}
Otherwise, if update $j$ is not delivered to a node $i$, the node waits for time $Y_j$ until the source generates the next update, and we refer to this random waiting time as a {\it service interval} for $\wqr$ nodes.
Suppose an update is delivered to node $i$ during service interval $j$ and the next successful update delivery to node $i$ is in service interval $j+M$. 
In this case, $M$ is a geometric r.v. with probability mass function (PMF) 
$P_M(m) = (1-p)^{m-1} p$, and first and second moments
\begin{align}
\E{M} & = \frac{1}{p} = \frac{n}{\wqr}, \nn
\E{M^2} & = \frac{2-p}{p^2} = \frac{2n^2}{\wqr^2} - \frac{n}{\wqr}. \label{eqn:EM}
\end{align}
We remark that $M$ and $Y_j$ are independent. 
An example of the age process is shown in Figure ~\ref{fig:sawtooth_w_1}. The update $j$ is delivered in service interval $j$ with end time $T_j$, and the node $i$ waits for $M=3$ service intervals until the next successful delivery in interval $j+3$. 
In Figure \ref{fig:sawtooth_w_1}, we represent each shaded trapezoid area as $A$ and the average age for earliest $\wqr$ stopping scheme with stopping threshold $\wqr$ is 
\begin{align}
\age_{(\wqr)} = \frac{\E{A}}{\E{M}\E{Y}}. \label{eqn:age_raw_w1}
\end{align}
Denote the random variable $\Xtil$ as the service time of a successful update delivered to some node $i$, we have $\text{E}\bigl[\Xtil\bigr]~=~\E{X_i|i\in\wset}$. 
Evaluating Fig.~\ref{fig:sawtooth_w_1} gives the area 
\begin{align}
	A_k = \frac{1}{2} \Bigl(\sum_{l=j}^{j+M_k-1} Y_l + \Xtil_{k}\Bigr)^2 - \frac{1}{2}\Xtil_{k}^2. \label{eqn:Ak}
\end{align}
Defining  $W=\sum_{l=j}^{j+M_k-1} Y_l$, the expected area of trapezoid
\begin{align}
\E{A} & = \frac{\E{W^2}}{2} + \E{W}\E{\Xtil}. \label{eqn:EA_w1}
\end{align}
Since $M$ and $Y$ are independent, $\E{W} = \E{M}\E{Y}$ and 
\begin{align}
	\E{W^2} = & = \E{M^2} (\E{Y})^2 + \E{M}\Var{Y}. \label{eqn:EW2}
\end{align}
Thus we have the following theorem
\begin{theorem}
For earliest $\wqr$ stopping scheme, the average age at an individual node is 
\footnote{This version is different from the published version in Allerton Conference 2017. An error correction has been made here based on the feedback from Prof.~Sennur~Ulukus at University of Maryland.}
\begin{align*}
\age_{(\wqr)} & = \frac{1}{\wqr}\sum_{i=1}^{\wqr} \E{X_{i:n}} + \frac{2n-\wqr}{2\wqr} \E{X_{\wqr:n}} + \frac{\Var{X_{\wqr:n}}}{2\E{X_{\wqr:n}}}. 
\end{align*}
\label{thm:age_w1}
\end{theorem}
\begin{proof}
It follows from \eqref{eqn:age_raw_w1}, \eqref{eqn:EA_w1} and \eqref{eqn:EW2} that
\begin{align}
\age_{(\wqr)} & = \E{\Xtil} + \frac{\E{M^2}\E{Y}}{2\E{M}} + \frac{1}{2} \frac{\Var{Y}}{\E{Y}} \nn
& = \E{X_j|j\in\wset} + \frac{\E{M^2} \E{X_{\wqr:n}}}{2\E{M}} + \frac{\Var{X_{\wqr:n}}}{2\E{X_{\wqr:n}}} \nn
& = \sum_{i=1}^{\wqr} \E{X_{i:n}} \Pr[j=i|i\in\wset] \nn
& \qquad \quad + \frac{\E{M^2} \E{X_{\wqr:n}}}{2\E{M}} + \frac{\Var{X_{\wqr:n}}}{2\E{X_{\wqr:n}}} \label{eqn:age_w1_proof3}.
\end{align}
In \eqref{eqn:age_w1_proof3}, the expected service time of a successful update $\text{E}\bigl[\Xtil\bigr]$ is obtained by averaging over the order statistics of service time $X_{i:n}$ from $i=1$ to $\wqr$. And the conditional probability
\begin{align}
\Pr[j=i|i\in\wset] = \frac{1/n}{\wqr/n} = \frac{1}{\wqr}. \label{eqn:cond_p_w1}
\end{align}
Substituting \eqref{eqn:EM} and \eqref{eqn:cond_p_w1} back to \eqref{eqn:age_w1_proof3} completes the proof.
\end{proof}
Similar to Theorem \ref{thm:age_n1}, Theorem \ref{thm:age_w1} is valid for any distribution of $X$ as long as the zero-wait policy is used and the update delivery acknowledgement is instantaneous. 
\begin{corollary}
Assuming $n$ is large and $n>\wqr$, we denote $\alpha=\wqr/n$.
For shifted exponential $(\lambda,c)$ service time $X$, the average age for earliest $\wqr$ stopping scheme can be approximated as
\begin{align*}
\age_{(\wqr)} \approx \hat{\age}(\alpha) =
\frac{1}{\lambda} - \frac{1}{2 \lambda} \log(1-\alpha) + \frac{\cf}{\alpha} + \frac{\cf}{2}. 
\end{align*}\label{thm:age_w1_approx}
\end{corollary}

\begin{proof}
For shifted exponentials $X$, we can rewrite the first term in Theorem \ref{thm:age_w1} as 
\begin{align}
\delta_1 & = \frac{1}{\wqr}\sum_{i=1}^{\wqr} \E{X_{i:n}} = \cf +\frac{1}{\wqr\lambda} \sum_{i=1}^{\wqr} (H_{n}-H_{n-i}).
\end{align}
Since the sum of Harmonic numbers has the the following series identity $\sum_{i=1}^{k} H_i = (k+1) (H_{k+1}-1)$, we have
\begin{align}
\delta_1 
& =  \cf + \frac{H_{n}}{\lambda} -  \frac{1}{\wqr\lambda} \left( \sum_{i=1}^{n-1} H_{i} - \sum_{i=1}^{n-\wqr-1} H_{i} \right) \nn
& =  \cf + \frac{H_{n}}{\lambda} -  \frac{1}{\wqr\lambda} \left( n (H_n -1) - (n-\wqr) (H_{n-\wqr} -1) \right) \nn
& =  \cf + \frac{1}{\lambda} - \frac{n-\wqr}{\lambda \wqr} (H_n-H_{n-\wqr}).
\end{align} 

With large $n$, we can again simply approximate the harmonic number by $H_i \approx \log i + \gamma $ and let $\wqr=\alpha n$. Thus, 
\begin{align}
\delta_1 & \approx  \cf + \frac{1}{\lambda} - \frac{n-\wqr}{\lambda \wqr} \log \left( \frac{n}{n-\wqr} \right) \nn
& = \cf + \frac{1}{\lambda} + \frac{1-\alpha}{\lambda \alpha} \log(1-\alpha). \label{eqn:delta_1}
\end{align}
Similarly, we denote the remaining terms in Theorem \ref{thm:age_w1} as
\begin{align}
\delta_2 
& = \frac{2n-\wqr}{2\wqr} \E{X_{\wqr:n}} + \frac{\Var{X_{\wqr:n}}}{2\E{X_{\wqr:n}}} \nn
& = \frac{2n-\wqr}{2\wqr} \left( \cf + \frac{H_{n}-H_{n-\wqr}}{\lambda}  \right) +  \frac{\Var{X_{\wqr:n}}}{2\E{X_{\wqr:n}}},
\end{align}
where the ratio 
\begin{align}
\frac{\Var{X_{\wqr:n}}}{\E{X_{\wqr:n}}} = \frac{H_{n^2}-H_{(n-\wqr)^2}}{2\lambda^2\cf + 2\lambda (H_n-H_{n-\wqr})}.
\end{align}
Remark that the sequence $H_{n^2}$ is monotonically increasing and it converges to $\pi^2/6$, thus $H_{n^2}-H_{(n-\wqr)^2}$ is relatively small and 
\begin{align}
\lim_{n\to\infty} \frac{\Var{X_{\wqr:n}}}{\E{X_{\wqr:n}}} = 0. \label{eqn:var_over_E}
\end{align}
This also implies $\E{X^2_{\wqr:n}}\approx (\E{X_{\wqr:n}})^2$ when $n$ is large.
Thus,
\begin{align}
\delta_2 
& \approx \frac{(2n-\wqr)\cf}{2\wqr} + \frac{2n-\wqr}{2\wqr\lambda} \left( H_{n}-H_{n-\wqr} \right) \nn
& \approx \frac{(2-\alpha)\cf}{2\alpha} + \frac{\alpha-2}{2\alpha\lambda} \log(1-\alpha). \label{eqn:delta_2}
\end{align}
The logarithm approximation of the average age is then followed by the summation of \eqref{eqn:delta_1} and \eqref{eqn:delta_2}.
\end{proof}
Note that this log approximation is tight only if $\wqr<n$, i.e. $\alpha<1$. As $\alpha$ approaches $1$, $\log n - \log(n-k) = -\log(1-\alpha)$ becomes unbounded, which differs from $H_n-H_{n-\wqr} = H_n$.

When $\cf=0$, the service time $X$ is exponential, and age is an increasing function of $\alpha$  for $\alpha\in(0,1)$ in Corollary~\ref{thm:age_w1_approx}.
Thus $\wqr=1$ is age minimized, implying that the source should send out a new update as soon as one node receives the current update.
This can be explained by the memoryless property of exponential random variables. 
For memoryless service, the residual service time required to deliver the current update to additional nodes is identical to the service time required to deliver a fresh update. 
However, the subsequent delivery of the fresh update will yield a larger reduction in age.
Thus the average age may be reduced if the wait for the current update is replaced by the wait for a new update. 

When $\cf>0$, by taking the derivative of Corollary \ref{thm:age_w1_approx} we can show that there exists a local minimum $\hat{\age}^* = \hat{\age}(\alpha^*) \leq \hat{\age}(\alpha)$ for $\alpha \in (0,1)$ in Corollary~\ref{thm:age_w1_approx}, which is given by  
\begin{align}
\alpha^* & = \sqrt[]{\lambda^2 \cf^2 + 2\lambda\cf} - \lambda\cf. \label{opt_alpha}
\end{align}
We note that given the number of nodes $n$, the optimal ratio $\alpha^*$ only depends on the product $\lambda \cf$. 
Given the number of nodes $n$, the near-optimal $\wqr^*$ can be obtained by rounding $\alpha^* n$ to the nearest integer, i.e. $\wqr^* = [\alpha^* n$].

\subsection{Stopping at Pre-selected $\wqr$}

In contrast to the earliest $\wqr$ stopping scheme which exploits all the feedback channels back to the source, we also provide the analysis for an alternative scheme, in which the source selected $\wqr$ of $n$ nodes randomly in advance and only waits for the response from these pre-selected $\wqr$ nodes.
For every update packet, a node is randomly pre-selected to the group with probability $p = k/n$.
The age process is similar to earliest $\wqr$ scheme as shown in Figure \ref{fig:sawtooth_w_1}, except that a particular user $i$ has a different probability $p^S$ for receiving an update, and the random service interval is reduced to $X_{\wqr:\wqr}$ since the source waits for the targeted $\wqr$ nodes. 
Similarly, we have the following theorem
\begin{theorem}
For pre-selected $\wqr$ stopping scheme, the average age at an individual node is 
\begin{align*}
\age^S_{(\wqr)} & = \frac{\wqr}{n}\E{X} + \frac{n-\wqr}{\wqr n}\sum_{i=1}^{\wqr}\E{X_{i:\wqr+1}} \nn
& \qquad \qquad + \frac{2n-k+nk}{2(k+nk)} \E{X_{\wqr:\wqr}} + \frac{ \Var{X^2_{\wqr:\wqr}}}{2\E{X_{\wqr:\wqr}}}. 
\end{align*}
\label{thm:age_w1S}
\end{theorem}
\begin{proof}
Since the age process follows similar patterns as the earliest $\wqr$ stopping scheme, the average can be obtained similarly from \eqref{eqn:age_raw_w1} and \eqref{eqn:EA_w1} as 
\begin{align}
\age^S_{(\wqr)} & = \E{\Xtil} + \frac{\E{M^2}\E{Y}}{2\E{M}} + \frac{1}{2} \frac{\Var{Y}}{\E{Y}}, \label{eqn:age_raw_S}
\end{align}
where the expected service interval $\E{Y} = \E{X_{k:k}}$ since the source only waits for the group of $k$ nodes.
Note that $\Xtil$ is now the service time given an individual node successfully receives the update, and $M$ is the geometric random variable with a different success probability $p^S$.

We again denote the set of pre-selected nodes as $\mathcal{K}$. 
If a particular node $i$ is pre-selected in the group of $\wqr$ nodes, this node receives the current update with probability $1$. 
On the other hand, if node $i$ is not in the pre-selected group, the link delay for node $i$ has to be smaller than the largest link delay in the pre-selected group.
Hence, the probability that an individual node $i$ gets the update is
\begin{align}
p^S & = p + (1-p) \Pr\{X_i<X_j, \; \exists j\in\mathcal{K}\} \nn
& = \frac{\wqr}{n} + \frac{n-\wqr}{n} \frac{\wqr}{1+\wqr}. \label{eqn:ps}
\end{align}
It follows from \eqref{eqn:EM} and \eqref{eqn:ps} that 
\begin{align}
\frac{\E{M^2}}{2\E{M}} & = \frac{2-p^S}{2p^S} \nn
& = \left(\frac{1}{2}\right) \frac{n+(n-\wqr)/(1+\wqr)}{n-(n-\wqr)/(1+\wqr)} \nn
& = \frac{2n-k+nk}{2(k+nk)} . \label{eqn:EMS}
\end{align}
Similarly, if node $i$ is in the pre-selected group, the conditional service time $\Xtil$ given that node $i$ gets the update is simply the service time $X_i$ since the delays are i.i.d. 
If node $i$ is not in the group, we compare the delay $X_i$ to the set of delay $\{X_j, \, j\in\mathcal{K}\}$. Given that node $i$ receives the update, $X_i$ cannot be the largest among all $k+1$ service times.
Thus,
\begin{align}
\E{\Xtil} & = p \E{\Xtil_i\,|\,i\in\mathcal{K}} + (1-p) \E{\Xtil_i\,|\,i \notin \mathcal{K}} \nn
& = \frac{\wqr}{n} \E{X} + \frac{n-\wqr}{n} \left(\frac{1}{\wqr} \sum_{i=1}^{\wqr} \E{X_{i:k+1}} \right). \label{eqn:EXtilS}
\end{align}
Substituting \eqref{eqn:EMS} and \eqref{eqn:EXtilS} back to \eqref{eqn:age_raw_S} completes the proof.
\end{proof}

\begin{corollary}
For shifted exponential $(\lambda,\cf)$ service time $X$, the average age for pre-selected $\wqr$ scheme is approximated by
\begin{align*}
\age^S_{(\wqr)} & \approx \cf+\frac{1}{\lambda}+\frac{n-\wqr}{\lambda\wqr n}(H_{\wqr+1}-1) \nn
& \qquad \qquad + \frac{2n-k+nk}{2(k+nk)} \left(\cf+\frac{H_k}{\lambda}\right). 
\end{align*}
\label{thm:age_w1S_approx}
\end{corollary}
Corollary \ref{thm:age_w1S_approx} follows directly by substituting the order statistics of shifted exponential r.v.s in \eqref{os_first} and \eqref{os_second}, and exploiting the property in \eqref{eqn:var_over_E} that $\Var{X_{k:n}}/\E{X_{k:n}}$ is negligible.

\section{Evaluation} \label{sec:evaluation}

\begin{figure}[t]
\centering
\includegraphics[width=0.5\textwidth]{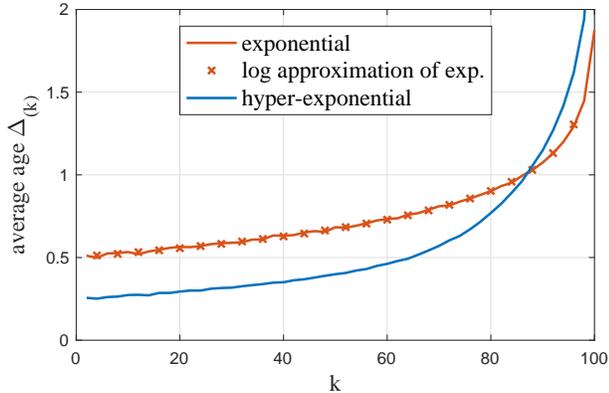}
\caption{Average age as a function of stopping threshold $\wqr$ using earliest $k$ scheme. Each link experiences exponential or hyper-exponential delay with the same expectation. $\times$ marks the log approximation of the average age for exponential service. }
\label{fig:exp}
\end{figure}

\begin{figure}[pt]
\centering
\includegraphics[width=0.5\textwidth]{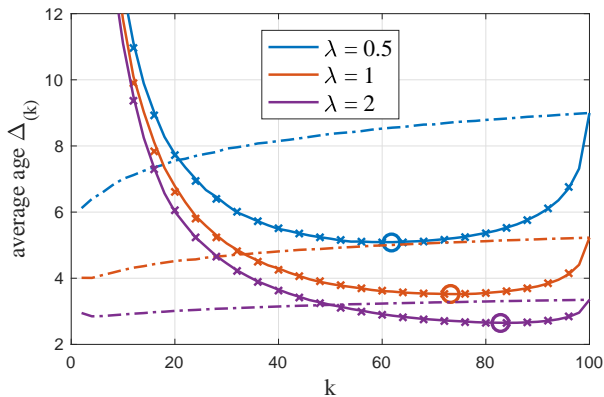}
\caption{Average age as a function of the stopping threshold $\wqr$ for shifted exponential service time. Solid and dash lines mark the average age for earliest $\wqr$ and pre-selected $\wqr$ stopping schemes respectively. $\times$ marks the approximation of the average age for earliest $k$ scheme. $\circ$ marks the minimized approximate age $\hat{\age}(\wqr^*)$.}
\label{fig:r1}
\negfigspace
\end{figure}

\begin{figure}[pt]
\centering
\includegraphics[width=0.5\textwidth]{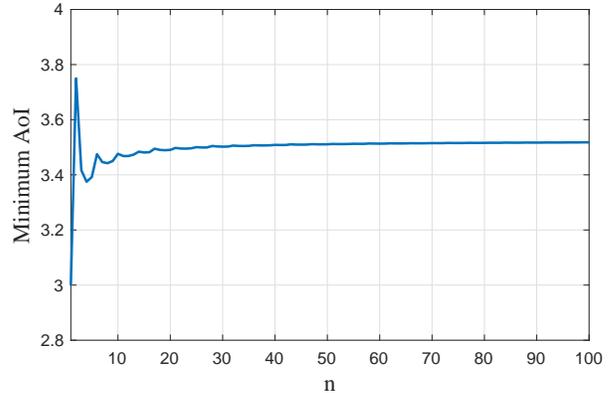}
\caption{Minimum average age obtained by optimizing the threshold $\wqr$ for different node numbers $n$ in the system with $\lambda=1$ and $\cf=1$.}
\label{fig:add_user}
\negfigspace
\end{figure}

Fig.~\ref{fig:exp} depicts the simulation of the average age $\age_{(\wqr)}$ as a function of the stopping threshold $\wqr$ using earliest k stopping scheme with $n=100$ nodes.
We compare the average age for exponential service time with rate $\lambda=2$, and hyper-exponential service time, which is the mixture of two exponentials with rates $\lambda_1=1$ and $
\lambda_2=6$ and $\lambda_1=1$ occurs with probability $p=0.4$. 
For $\cf=0$, Corollary~\ref{thm:age_w1_approx} becomes 
\begin{align*}
\hat{\age}(\alpha)~=~\frac{1}{\lambda}+\frac{\log(1-\alpha)}{2\lambda},
\end{align*}
which implies that the normalized age $\lambda \age_{(\wqr)}$ is almost independent of the rate $\lambda$.
We observe that the average age increases monotonically as $\wqr$ increases for both exponential and hyper-exponential $X$, while the logarithmic term in Corollary~\ref{thm:age_w1_approx} provides a tight approximation to the true average age with exponential service time.
We also observe in Fig.~\ref{fig:exp} that the hyper-exponential service results in a steeper curves than the exponential, mainly because of the {\it log-convex} behavior of the hyper-exponential distribution.
For log-convex $X$, the tail probability satisfies $\Pr(X>x+t|X>t) \geq \Pr(X >x)$, which indicates that the more you wait for the completion of event $X$, the more likely you will wait even longer. 
Thus, having a smaller stopping threshold  will preempt the current update earlier, which lowers the average age. Thus the optimal strategy is to stop waiting as early as possible, which yields $\wqr^*=1$.

Similarly, Figure \ref{fig:r1} compares the average age $\age_{(\wqr)}$ for the earliest $k$ and pre-selected $k$ stopping schemes, where each link has shifted exponential transmission time with time shift $c=1$ and different rate $\lambda$. 
The logarithmic approximation in Corollary~\ref{thm:age_w1_approx} for earliest $k$ scheme is marked with $\times$, and the near-optimal stopping threshold  $\wqr^*$ is marked with $\circ$. 
We observed that for a given number of nodes $n$, and the optimal stopping threshold  $\wqr^*$ increases as the product $\lambda c$ increases, which also indicates the minimum of the actual average age obtained by simulation.  
For instance, when $\lambda c=0.5$, the optimal strategy is to wait for more than 60 responses from the nodes before sending a new update.

We also observe that the earliest $\wqr$ scheme provides lower average age than the pre-selected $\wqr$ scheme by choosing the threshold $\wqr$ properly, and both schemes give the same average age when $\wqr=n$. 
For pre-selected $\wqr$ scheme with large $n$ and very small threshold $\wqr$, e.g. $\wqr=1$, an individual node is very unlikely to be pre-selected in the group of $k=1$ node, and this node has about half of the chance to be faster than the chosen node and gets an update. That is, the source will receive responses from around half of the $n$ nodes while it's waiting for the pre-selected node.
As $\wqr$ increases, the probability that a particular node is pre-selected is increased, and every unselected node has larger probability to be faster than the slowest node in the pre-selected group. This implies that the source actually waits for significantly more nodes than $\wqr$ as the threshold $\wqr$ increases for the pre-selected $\wqr$ scheme.

Figure \ref{fig:add_user} depicts the average age at an individual node as the total number of nodes $n$ in the system increases. For every $n$, a corresponding age-minimized stopping threshold $k^*$ is obtained by \eqref{opt_alpha}. 
When there is only one node in the system, we observe that the average age is exactly at $3(\cf+\frac{1}{\lambda})/2=3$. 
The average age fluctuates for small $n$ since the earliest $\wqr$ scheme is a heuristic by waiting for an integer number of users $\wqr$. 
As the number of users $n$ grows, the average age converges rapidly, which matches our analysis in Corollary \ref{thm:age_w1_approx} that the average age depends only on $\lambda$ and $\cf$ when $n$ is large enough. 
This implies the effect of increasing number of users in a share network on the average age can be eliminated by choosing the multicast strategy properly.

\section{Conclusion and Future Work} \label{sec:conclusion}
We have examined a status updates multicast system in which real-time status update messages are replicated and sent to a set of nodes.
The source preempts the current update with a new update as soon as it receives the acknowledgements back from any $\wqr$ out of $n$ nodes.
The freshness of the updating system is measured by the time-averaged age at each single node.
As the stopping threshold $\wqr$ increases, each node is more likely to get the most recent update generated by the source. 
However, the age of the content, which represents the freshness in time, also increases as $\wqr$ increases for exponential link delay.
For shifted exponential service model, we have derived the age optimized stopping threshold $\wqr$ that balances the likelihood of getting the last generated update and the freshness of that update message.  
We have shown numerically that the earliest $\wqr$ stopping scheme outperforms the pre-selected $\wqr$ scheme by choosing the age minimized threshold $\wqr$. 
The observations for different services distributions also lead us to believe that the tradeoff between the stopping threshold and the age of information is dictated by the heaviness of the tail of the service time distribution.

In this work, we assumed the feedback channels from the receivers are instantaneous, and a more generalized model with limited feedback from the receivers will be addressed in future work.
Our results are limited to the scenario where the update transmission is incremental in time such that the transmission can be terminated at any time, but we are also interested in an alternative scenario where new packets are queued in the network because old packets are still in the service.
A extensive characteristic of the tradeoff between network resources and age of information is also of our interest.

\section*{Acknowledgment}
\addcontentsline{toc}{section}{Acknowledgment}

The authors would like to thank Prof. Sennur Ulukus and her research group at University of Maryland for constructive feedback of the manuscript. 

Part of this research is based upon work supported by the National Science Foundation under Grant No.~CIF~-~1422988.

\balance

\bibliographystyle{IEEEtran}
\bibliography{ref}

\end{document}